\documentclass[onecolumn, accepted=2017-06-26]{quantumarticle}

\usepackage[utf8]{inputenc}
\usepackage[english]{babel}
\usepackage[T1]{fontenc}

\usepackage{tikz}
\usepackage{lipsum}

\usepackage{enumerate}
\usepackage{amsmath}
\usepackage{amssymb}
\usepackage{graphicx}
\usepackage{color}
\usepackage{hhline}
\usepackage{tikz}
\usetikzlibrary{positioning}
\usetikzlibrary{backgrounds,fit,decorations.pathreplacing} 
\usepackage[hyperref,amsthm,amsmath,thmmarks]{ntheorem}
\usepackage{fullpage}
\allowdisplaybreaks[1]
\usepackage[numbers,sort&compress]{natbib}
\usepackage{hyperref}

\theorembodyfont{\rmfamily}
\newtheorem{theorem}{\textit{Theorem} }
\newtheorem{proposition}[theorem]{\textit{Proposition}}
\newtheorem{definition}[theorem]{\textit{Definition} }
\newtheorem{corollary}[theorem]{\textit{Corollary}}
\newtheorem{lemma}[theorem]{\textit{Lemma}}

\newcommand{\lv}{\left \vert}
\newcommand{\rv}{\right \vert}

\newcommand{\ra}{\right \rangle}
\newcommand{\ket}[1]{\lv #1 \ra}

\newcommand{\braket}[2]{\langle #1 \vert #2 \rangle}
\newcommand{\ketbra}[2]{\lv #1 \rangle \langle #2 \rv}

\newcommand{\tr}{\mathrm{Tr}}
\newcommand{\mc}[1]{\mathcal{#1}}

\title{Decoupling with random diagonal unitaries}

\author{Yoshifumi Nakata}
\affiliation{Institut f\"ur Theoretische Physik, Leibniz Universit\"at Hannover,  Appelstrasse 2, 30167 Hannover, Germany.}
\affiliation{Photon Science Center, Graduate School of Engineering, The University of Tokyo, Bunkyo-ku, Tokyo 113-8656, Japan.}
\affiliation{Departament de F\'{\i}sica: Grup d'Informaci\'{o} Qu\`{a}ntica, Universitat Aut\`{o}noma de Barcelona, ES-08193 Bellaterra (Barcelona), Spain.}
\email{nakata@qi.t.u-tokyo.ac.jp}

\author{Christoph Hirche}
\affiliation{Institut f\"ur Theoretische Physik, Leibniz Universit\"at Hannover,  Appelstrasse 2, 30167 Hannover, Germany.}
\affiliation{Departament de F\'{\i}sica: Grup d'Informaci\'{o} Qu\`{a}ntica, Universitat Aut\`{o}noma de Barcelona, ES-08193 Bellaterra (Barcelona), Spain.}
\email{christoph.hirche@uab.cat}

\author{Ciara Morgan}
\affiliation{Institut f\"ur Theoretische Physik, Leibniz Universit\"at Hannover,  Appelstrasse 2, 30167 Hannover, Germany.}
\affiliation{
School of Mathematics and Statistics, University College Dublin, Belfield, Dublin 4. Ireland.}
\email{ciara.morgan@ucd.ie}

\author{Andreas Winter}
\affiliation{Departament de F\'{\i}sica: Grup d'Informaci\'{o} Qu\`{a}ntica, Universitat Aut\`{o}noma de Barcelona, ES-08193 Bellaterra (Barcelona), Spain.}
\affiliation{
ICREA--Instituci\'{o} Catalana de Recerca i Estudis Avan\c{c}ats, Pg. Lluis Companys, 23, ES-08010 Barcelona, Spain}
\email{andreas.winter@uab.cat}
\begin{document}

\maketitle

\begin{abstract}
We investigate decoupling, one of the most important primitives in quantum Shannon theory, by replacing the uniformly distributed random unitaries commonly used to achieve the protocol, with repeated applications of random unitaries diagonal in the Pauli-$Z$ and -$X$ bases. This strategy was recently shown to achieve an approximate unitary $2$-design after a number of repetitions of the process, which implies that the strategy gradually achieves decoupling. Here, we prove that even fewer repetitions of the process achieve decoupling at the same rate as that with the uniform ones, showing that rather imprecise approximations of unitary $2$-designs are sufficient for decoupling. We also briefly discuss efficient implementations of them and implications of our decoupling theorem to coherent state merging and relative thermalisation.
\end{abstract}

\section{Introduction}

The decoupling theorem states that the task of destroying all possible correlations between two quantum systems, can be achieved by applying a randomly chosen unitary to one of the systems, with error bounds expressed in terms of entropic correlation measures.  Versions of the theorem can be traced back to the origins of the field of quantum Shannon theory~\cite{SW2002,D2005,DW2004,GPW2005} and the original form is responsible not only for simplifying the derivation of known quantum protocols \cite{ADHW2009, DH2011}, offering critical insight into these procedures, but also for providing operational interpretations for entropic quantities and correlation measures.  
The decoupling theorem, or simply \emph{decoupling}, has been particularly used to achieve the state merging protocol~\cite{HOW07}, giving rise to an operational interpretation for the negativity of quantum conditional entropy that represents a fundamental deviation from the classical counterpart. 
Decoupling is also relevant to fundamental physics in complex quantum systems, such as the information paradox of quantum black holes~\cite{HP2007} and thermalisation in isolated quantum systems~\cite{dRHRW2014}.

Traditionally, the unitary applied during the protocol is chosen uniformly at random over the entire unitary group according to the unique unitarily invariant probability measure, also known as the Haar measure. These unitary operators are often referred to as Haar random unitaries. In Ref.~\cite{DBWR2010,DupuisThesis}, decoupling with Haar random unitaries was studied in its most general form, and it was revealed that whether this protocol can be achieved depends on a simple sum of the smooth entropy of the initial state and that of a state corresponding to the quantum channel involved in the protocol. 
Despite that this result has wide applications in quantum information science, Haar random unitaries cannot be efficiently implemented by quantum circuits. 
Here \emph{efficiently} means by using at most a polynomial number of gates and random bits. 
This fact raises the question of what kinds of efficiently implementable probability measures on the unitary group achieve decoupling as strong as the Haar measure.
One of the analyses goes through using so-called \emph{unitary $t$-designs}~\cite{DLT2002,DCEL2009,TGJ2007,GAE2007,BWV2008a,WBV2008,HL2009,DJ2011,L2010,BHH2012,CLLW2015,NHKW2016}, which simulate up to the $t$th order statistical moments of Haar random unitaries,
and it turned out that approximations of unitary $2$-designs can achieve decoupling if the approximation is precise enough~\cite{SDTR2013}.
However, it remains open whether such precise approximations are necessary or not.

In this paper, we propose a new method of achieving decoupling, based on the repeated application of random unitaries diagonal in the Pauli-$Z$ and -$X$ bases.
This process was first introduced in Ref.~\cite{NHMW2015-1} and has been shown to constitute an approximate unitary $2$-design, which was later proven to constitute $t$-designs for general $t$~\cite{NHKW2016}.
Together with the fact that an approximate 2-design can achieve decoupling~\cite{SDTR2013}, it immediately follows that the process achieves the decoupling as strong as Haar random ones after a certain number of repetitions.
However, we here go on and show that such strong decoupling is achievable after even fewer repetitions,
implying that rather imprecise approximations of unitary $2$-designs achieve decoupling at the same rate as the Haar one.
This result may open the possibility that decoupling could be achieved by random unitaries with less structure than unitary $2$-designs even if we require the same rate as the Haar random one.
We also prove the concentration of measure for this process and briefly discuss implications of our results to coherent state merging protocols~\cite{HOW07} and relative thermalisation~\cite{dRHRW2014}.

For an introduction to the decoupling theorem in the asymptotic setting we refer, in particular, to a tutorial by Hayden~\cite{HaydenTutorial} and in the non-asymptotic setting we refer to Ref.~\cite{DupuisThesis}. An introduction to the uses of random diagonal unitaries in quantum information science can be found in Ref.~\cite{NM2014}.
 
The article is organised as follows. In Section \ref{Prelims} we begin by introducing the necessary notation, where three known versions of the decoupling theorem are also presented for the sake of comparison. 
The main results are presented in Section~\ref{MainResult}. There, we also briefly discuss efficient implementations and applications.
Section~\ref{Proofs} contains the proofs of our results and we conclude with Section~\ref{Conclusion}. 

\section{Preliminaries} \label{Prelims}

We first introduce our notation, including the definitions of several entropies and random unitaries, in Subsection~\ref{SS:Notation}.
We then provide three known versions of decoupling theorem in Subsection~\ref{SS:D}.

\subsection{Notation} \label{SS:Notation}

We use the standard asymptotic notations, such as $O$, $\Omega$, and $\Theta$, 
and the standard norms for operators and superoperators, such as the Schatten $p$-norms and the diamond norm.
The definitions are listed in Appendix~\ref{App:ND}.

Let $A$ be a system composed of $N_A$ qubits. We denote by $\mc{H}_A$ the corresponding Hilbert space and by $d_A= 2^{N_A}$ the dimension of $\mc{H}_A$.
The set of bounded operators, positive semidefinite operators, subnormalised states and states on a Hilbert space $\mc{H}$ are denoted by
$\mc{B}(\mc{H})$, $\mc{P}(\mc{H})$, $\mc{S}_{\leq}(\mc{H}):=\{ \rho \in \mc{P}(\mc{H})| \tr \rho \leq 1 \}$, and $\mc{S}(\mc{H}):=\{ \rho \in \mc{P}(\mc{H})| \tr \rho = 1 \}$, respectively.
Reduced operators of $X_{AR} \in \mc{B}(\mc{H}_{A R})$ are denoted by $X_A:= \tr_R X_{AR}$ and $X_R := \tr_A X_{AR}$.
Throughout the paper $\rho^{-1}$ will be the Moore-Penrose generalised inverse of $\rho$. 
We use the Choi-Jamio\l kowski representation $J(\mc{T}_{A \rightarrow B})$ of a linear map $\mc{T}_{A \rightarrow B}:\mc{B}(\mc{H}_A) \rightarrow \mc{B}(\mc{H}_B)$, which is indeed an isomorphism (see Appendix~\ref{App:AL}).
In quantum mechanics, an important class of liner maps is a set of completely positive and trace preserving (CPTP) maps, often referred to as \emph{quantum channels}, of which Choi-Jamio\l kowski representations are sometimes called the Choi-Jamio\l kowski states.

We also use entropic quantities, namely the \emph{conditional collision entropy} and \emph{conditional min-entropy}.
\begin{definition}[Conditional collision entropy and conditional min-entropy~\cite{RennerThesis, T16}]
Let \newline $\rho_{AB} \in \mc{S}_{\leq}(\mc{H}_{AB})$. \emph{The conditional collision entropy} $H_2(A|B)_{\rho}$ and \emph{the conditional min-entropy} $H_{\rm min}(A|B)_{\rho}$ of $A$ given by $B$ are defined by, respectively,
\begin{align}
&H_2(A|B)_{\rho} = \sup_{\sigma_B \in \mc{S} (\mc{H}_B)} - \log \tr \bigl( (I_A \otimes \sigma_B^{-1/4}) \rho_{AB} (I_A \otimes \sigma_B^{-1/4}) \bigr)^2,\\
&H_{\rm min}(A|B)_{\rho} = \sup_{\sigma_B \in \mc{S}(\mc{H}_B)} \sup \{ \lambda \in \mathbb{R} : 2^{-\lambda} I_A \otimes \sigma_B - \rho_{AB} \geq 0 \},
\end{align}
where $I_A$ is the identity operator on $\mc{H}_A$.
\end{definition}
%These entropies satisfy
%\begin{equation}
%- \min\{\log d_A, \log d_B\} \leq H_{\rm min}(A|B)_{\rho} \leq H_2(A|B)_{\rho}  \leq \log d_A,\label{Eq:entropies}
%\end{equation}
%for $\rho_{AB} \in \mc{S}(\mc{H}_{AB})$~\cite{T16}

We also use a smoothed version of the conditional min-entropy.
Let $P(\rho, \sigma)$ be the purified distance for $\rho, \sigma \in S_{\leq}(\mc{H})$, given by $P(\rho, \sigma) = \sqrt{1 - \bar{F}(\rho, \sigma)^2}$ where $\bar{F}(\rho, \sigma) :=  |\!| \sqrt{\rho} \sqrt{\sigma}  |\!|_1 -\sqrt{(1- \tr \rho)(1- \tr \sigma)}$.

\begin{definition}[{\bf Smooth conditional min-entropy~\cite{RennerThesis}}]
Let $\epsilon \geq 0$ and $\rho_{AB} \in S_{\leq}(\mc{H}_{AB})$. The \emph{$\epsilon$-smooth conditional min-entropy} of $A$ given $B$ is defined by
\begin{equation}
H^{\epsilon}_{\rm min}(A|B)_{\rho} = \sup_{\substack{ \sigma \in S_{\leq}(\mc{H}_{AB}), \\ P(\rho_{AB}, \sigma) \leq \epsilon}} H_{\rm min}(A|B)_{\sigma}.
\end{equation}
\end{definition}

Random unitaries play crucial roles in decoupling. The definitions are given below.

\begin{definition}[Haar random unitaries]
{\it Let $\mc{U}$ be the unitary group, and denote the Haar measure (i.e. the unique unitarily invariant probability measure) 
 on $\mc{U}$ by ${\sf H}$. A \emph{Haar random unitary}
 $U$ is a $\mc{U}$-valued random variable distributed according to the Haar measure, $U \sim {\sf H}$.}
\end{definition}

\begin{definition}[{\bf Random $X$- and $Z$-diagonal unitaries~\cite{NM2013}}]
{\it Let $\mathcal{D}_{W}(d)$ be the set of unitaries diagonal in the Pauli-$W$ basis $\{ \ket{n}_W \}_{n=0}^{d-1}$ ($W=X,Z$), given by $\bigl\{\sum_{n=0}^{d-1} e^{i \varphi_n} \ketbra{n}{n}_W :\varphi_n \in [0, 2\pi)  \text{  for  }  n \in [0,\ldots,d-1]  \bigr\}$. 
Let ${\sf D}_W$ denote a probability measure on it induced by a uniform probability measure on its parameter space $[0,2 \pi)^d$. A \emph{random $W$-diagonal-unitary} $D$ is a $\mathcal{D}_{W}(d)$-valued random variable distributed according to ${\sf D}_W$, $D \sim {\sf D}_W$.}
\end{definition}

These random unitaries have been applied to a wide variety of problems in quantum information science (see e.g.~\cite{L2010} and~\cite{NM2014}).
%~\cite{BR2003,AMTW2000,AS2004,DN2006,HLSW2004,Au2009,L2000,HHL2004,DLT2002,TDL2001,L1997,RBSC2004,RRS2005,S2006,DCEL2009}
%and have been used to investigate typical properties in physical systems~\cite{GLTZ2006,PSW2006,R2008,HP2007,dRHRW2014,NTM2012}.
However, implementations of Haar random unitaries on $N$ qubits necessarily use $\Omega(2^{2N})$ quantum gates and at least the same amount of random bits~\cite{K1995}, and so, cannot be efficient.
This fact has led to the investigations of {\it unitary $t$-designs}~\cite{DLT2002,DCEL2009,TGJ2007,GAE2007,BWV2008a,WBV2008,HL2009,DJ2011,L2010,BHH2012,CLLW2015}. 
Let $t$ be a natural number and $\mathcal{G}_{U \sim \nu}^{(t)}(X)$ be a CPTP map given by $\mathcal{G}_{U\sim \nu}^{(t)}(X) := \mathbb{E}_{U \sim \nu} [ U^{\otimes t} X U^{\dagger \otimes t}]$ for any $X \in \mc{B}(\mc{H}^{\otimes t})$, where $\mathbb{E}_{U \sim \nu}$ represents an average over a probability measure $\nu$.
A $\delta$-approximate unitary $t$-design is then defined as follows.
\begin{definition}[Unitary $\boldsymbol{t}$-designs~\cite{DCEL2009,HL2009}] \label{Def:Ut}
Let $\nu$ be a probability measure on the unitary group.
A random unitary $U \sim \nu$ is called a \emph{$\delta$-approximate unitary $t$-design} 
 if $|\!|  \mathcal{G}^{(t)}_{U \sim \nu} - \mathcal{G}^{(t)}_{U \sim {\sf H}}  |\! |_{\diamond} \leq \delta$.
The design is called \emph{exact} when $\delta=0$.
\end{definition}
%Since the Haar measure is a uniform distribution in the unitary group, the $t$ for a unitary $t$-design is an indicator of how uniform the distribution is.
%Recalling that the diamond norm provides the optimal success probability of distinguishing two quantum channels, a $\delta$-approximate unitary $t$-design cannot be distinguished up to error $\delta$ from a Haar random one even if we have $t$ copies of the unitary.

Finally, we introduce a probability measure related to a random quantum circuit (RQC).
An RQC, first introduced in Refs.~\cite{ODP2007,DOP2007}, is a quantum circuit where each local gate is chosen independently and uniformly at random from a given gate set and is applied to a randomly chosen pair of qubits. 
In this paper, we denote the RQCs, where the gate set is chosen to be a set of all two-qubit gates, with length $O(L)$ by $RQC(L)$.
The $RQC(L)$ is naturally equipped with a probability measure, induced from the random choices of gates and qubits on which the gates are applied. We denote the measure by ${\sf RQC}(L)$.

\subsection{Generalised decoupling: Setting and previous results} \label{SS:D}

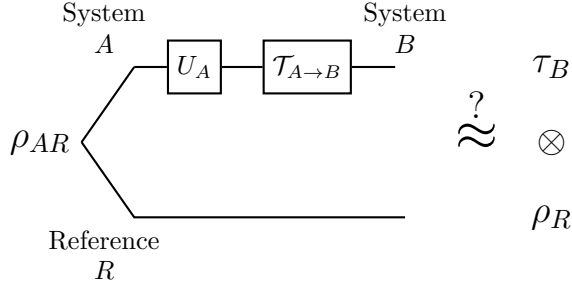
\begin{figure}[tb!]
\centering
   \begin{tikzpicture}[thick]
    %
    % `operator' will only be used by Hadamard (H) gates here.
    % `phase' is used for controlled phase gates (dots).
    % `surround' is used for the background box.
    \tikzstyle{operator1} = [draw,fill=white,minimum size=1.5em, minimum width=3em] 
    \tikzstyle{operator} = [draw,fill=white,minimum size=2em] 
    \tikzstyle{phase} = [fill,shape=circle,minimum size=5pt,inner sep=0pt]
    %\tikzstyle{surround} = [fill=blue!10,thick,draw=black,rounded corners=2mm]
    %
    \node at (0,0) (origin) {};
    %
    % Column 1
    \path[draw] (0,0) -- (0.7,1) (0,0) -- (0.7,-1) -- (4.3,-1);
	\node[operator] (U) at (1.5,1) {$U_A$} edge [-] (0.7,1);
	\node[operator] (T) at (3,1) {$\mc{T}_{A \rightarrow B}$} edge [-] (U);
	\node at  (4.3,1) {}  edge [-] (T);
	\node at (0.3,1.7) {System};
	\node at (0.3,1.3) {$A$};
	\node at (0.3,-1.3) {Reference};
	\node at (0.3,-1.7) {$R$};
	\node at (-0.55,0) {\Large $\rho_{AR}$};
	\node at (4.3,1.7) {System};
	\node at (4.3,1.3) {$B$};

%	\node at (5.2,0) {\huge  $ = $};
%
%    \path[draw] (7,0) -- (7.7,1) (7,0) -- (7.7,-1);
%	\node[operator] (U) at (8.5,1) {$U_A$} edge [-] (7.7,1);
%	\node[operator] (T) at (10,1) {$\mc{T}_{A \rightarrow B}$} edge [-] (U);
%	\node[operator] (E) at (9.5,-1) {$\mc{E}_{\bar{A} \rightarrow R}$} edge [-] (7.7,-1);
%	\node at  (11.3,1) {}  edge [-] (T);
%	\node at  (11.3,-1) {}  edge [-] (E);
%	\node at (7.3,1.7) {System};
%	\node at (7.3,1.3) {$A$};
%	\node at (7.3,-1.3) {System};
%	\node at (7.3,-1.7) {$\bar{A}$};
%	\node at (6.45,0) {\Large $\Phi_{A \bar{A}}$};
%	\node at (11.3,1.7) {System};
%	\node at (11.3,1.3) {$B$};
%	\node at (11.3,-1.3) {Reference};
%	\node at (11.3,-1.7) {$R$};

	\node at (5.25,0.3) {\huge  $ \stackrel{?}{\approx} $};
	\node at (6.25,1) {\Large  $\tau_{B}$};
	\node at (6.25,0) {\Large  $\otimes$};
	\node at (6.25,-1) {\Large  $\rho_R$};
%    \draw[-] (op21) at (1,-1) {$\varphi_2$} edge [-] (origin);
    \end{tikzpicture}
\caption{This figure depicts the decoupling protocol in the most general form. The left diagram shows the protocol starting with an initial state $\rho_{AR}$. The goal is to make the output state $\mc{T}_{A \rightarrow B}(U_A \rho_{AR} U_A^{\dagger})$ as close as possible to $\tau_B \otimes \rho_{R}$ by choosing an appropriate unitary $U_A$.}
\label{Fig:dec0}
\end{figure}

%To explain the decoupling protocol in the most general form,
Let us consider the following bipartite setting with a system $A$ and a reference $R$.
Initially, they share a given initial state $\rho_{AR} \in \mathcal{S}(\mathcal{H}_{AR})$.
We apply a unitary $U_A$ on the system $A$, and the state on $A$ is sent to $B$ via a given quantum channel $\mathcal{T}_{A\rightarrow B}$.
The task of decoupling is to make the resulting state $\mc{T}_{A \rightarrow B}  (U_A \rho_{AR} U_A^{\dagger})$ as close to a product state as possible by selecting an appropriate unitary $U_A$ (see also Fig.~\ref{Fig:dec0}).
In particular, when the resulting state is approximately $\tau_B \otimes \rho_R$, where $\tau_B = \tr_{A} J(\mc{T}_{A \rightarrow B})$ and $\rho_R= \tr_A \rho_{AR}$,
it is often said that {\it generalised decoupling} is achieved~\cite{DBWR2010,DupuisThesis}. We consider this generalised version in this paper and simply refer to it as decoupling.
%It is obvious that decoupling is not always achievable. For instance, when the initial state $\rho_{AR}$ is a maximally entangled state and the CPTP map is given by the identity map, the output state is also maximally entangled between $B$ and $R$, implying that decoupling cannot be achieved.
%Hence, the main goal in the study of decoupling is to answer the questions of when and, if it is achievable, how decoupling can be achieved.

In this paper, for a given probability measure $\nu$ on the unitary group,
we say decoupling is achievable with an error $\Delta$ when
\begin{equation}
\mathbb{E}_{U_A \sim \nu}|\!| \mc{T}_{A \rightarrow B}(U_A \rho_{AR} U_A^{\dagger}) - \tau_B \otimes  \rho_R |\!|_1 \leq \Delta. \label{Eq:kn42}
\end{equation} 
At this point we can draw an analogy to the asymptotic setting of many independent and identically distributed (i.i.d.) states where we would define the decaying rate of the error as $\lim_{n\rightarrow\infty}\frac{-1}{n}\log\Delta$. For our purposes, if $\Delta = O(2^{-\Lambda})$, in a slight abuse of nomenclature, we call $\Lambda$ the \emph{decoupling rate} with the random unitary $U_A$. This analogy is justified in the sense that all errors $\Delta$ leading to the same decoupling rate $\Lambda$ give the same asymptotic rate when applying the protocol to many i.i.d. states.   
Since the left-hand side of Eq.~\eqref{Eq:kn42} is non-negative, when the error $\Delta$ is sufficiently small or the decoupling rate $\Lambda$ is sufficiently large,
almost all unitaries according to the measure $\nu$ approximately achieve decoupling.

In the literature, three generalised versions of decoupling are known, which we summarise in the following theorem.
\begin{theorem}[Decoupling theorems~\cite{DBWR2010,SDTR2013,BF2013}] \label{Thm:decs}
Let ${\sf H}$, and ${\sf Des^{(2,{\delta})}}$ be the Haar measure and a $\delta$-approximate unitary $2$-design, respectively. 
Then, for $\rho_{AR} \in \mc{S}(\mc{H}_{AR})$, a CPTP map $\mc{T}_{A \rightarrow B}:\mc{B}(\mc{H}_A) \rightarrow \mc{B}(\mc{H}_B)$,
and $\epsilon > 0$,
\begin{align}
& \mathbb{E}_{U_A \sim {\sf H}}|\!| \mc{T}_{A \rightarrow B}(U_A \rho_{AR} U_A^{\dagger}) - \tau_B \otimes  \rho_R |\!|_1 \leq 2^{-\frac{1}{2} (H^{\epsilon}_{\rm min}(A|R)_{\rho}+H^{\epsilon}_{\rm min}(A|B)_{\tau})}+12 \epsilon, \label{Eq:decHaar}\\
&\mathbb{E}_{U_A \sim {\sf Des^{(2,{\delta})}}}|\!| \mc{T}_{A \rightarrow B}(U_A \rho_{AR} U_A^{\dagger}) - \tau_B \otimes  \rho_R |\!|_1
\leq 
\sqrt{1+4 \delta d_A^4}\  2^{-\frac{1}{2} (H^{\epsilon}_{\rm min}(A|R)_{\rho}+H^{\epsilon}_{\rm min}(A|B)_{\tau} )} + 8 d_A \delta \epsilon + 12 \epsilon. \label{Eq:dec2des}
\end{align}
Moreover, for any $\eta>0$, an RQC($L$) satisfies
\begin{align}
&\mathbb{E}_{U_A \sim {\sf RQC}(L)}|\!| \mc{T}_{A \rightarrow B}(U_A \rho_{AR} U_A^{\dagger}) - \tau_B \otimes  \rho_R |\!|_1
\leq 
\sqrt{\frac{1}{{\rm poly}(N_A)} +  2^{2\eta N_A- H_{2}(A|R)_{\rho}-H_{2}(A|B)_{\tau}}}, \label{Eq:decRQC}
\end{align}
for all $L \geq c N_A (\log N_A)^2$, where $c$ is a constant.
\end{theorem}
A decoupling rate with a Haar random unitary is given from Ineq.~\eqref{Eq:decHaar}: for sufficiently small $\epsilon$, 
\begin{equation}
\Lambda_{\rm Haar}=\frac{1}{2} (H^{\epsilon}_{\rm min}(A|R)_{\rho}+H^{\epsilon}_{\rm min}(A|B)_{\tau}).
\end{equation}
This rate can be also achieved by $\delta$-approximate unitary $2$-designs when $\delta=O(d_A^{-4})$, which can be observed from Eq.~\eqref{Eq:dec2des}. 
Due to the asymptotic equipartition property of the smooth conditional min-entropy~\cite{TCR2009}, $\Lambda_{\rm Haar}$ can be easily transformed to the standard rate in the asymptotic limit of the i.i.d. setting, given by the sum of the conditional von Neumann entropies.
The inequality~\eqref{Eq:decRQC} shows that $RQC\bigl(N_A(\log N_A)^2 \bigr)$ achieves decoupling at the expense of exponentially small decoupling errors.

We note that no random unitaries are known so far, which are worse approximations than $O(d_A^{-4})$ of unitary $2$-designs but still achieve decoupling at the same rate as $\Lambda_{\rm Haar}$. 
This motivates us to explore imprecise approximations of unitary $2$-designs achieving decoupling at the rate $\Lambda_{\rm Haar}$.

\section{Main results} \label{MainResult}

Here, we show that decoupling at the rate $\Lambda_{\rm Haar}$ can be achieved by $\Theta(d_A^{-2})$-approximate unitary $2$-designs.
We then briefly discuss efficient implementations, and implications to coherent state merging~\cite{HOW07} and relative thermalisation~\cite{dRHRW2014}.\\

Our strategy is to use independent random $Z$- and $X$-diagonal unitaries. 
As it is obvious that one use of a random $Z$- or $X$-diagonal unitary cannot achieve decoupling, we need to consider combinations of them, such as alternate applications of random $Z$- and $X$-diagonal unitaries.
It turns out however that one repetition of them does not achieve decoupling in general, which is presented in Proposition~\ref{Prop:example} below (see Appendix~\ref{Ex:Proof} for the proof).

\begin{proposition} \label{Prop:example}
Consider a system $A$ composed of two subsystems $A_1$ and $A_2$ with dimension $d_{A_1}$ and $d_{A_2}$, respectively. 
Let $\rho_{AR}$ be $\rho_{AR}=\Phi_{A_1 R} \otimes \ketbra{0}{0}_{A_2}$, where $\Phi_{A_1 R}$ is the maximally entangled state between $A_1$ and $R$, and $\ket{0}$ is in the Pauli-$Z$ basis.
Then, for a random unitary $U_A = D^X D^Z$ ($D^X \sim {\sf D}_X$ and $D^Z \sim {\sf D}_Z$), we have
\begin{equation}
\mathbb{E}_{U_A \sim {\sf D}_X \times {\sf D}_Z}|\!|  \tr_{A_2} (U_A  \rho_{AR}  U_A^{\dagger} )- \tr_{A}J(\tr_{A_2}) \otimes  \rho_R |\!|_1 
\geq 
\frac{1}{\sqrt{2}}
\biggl(\frac{1}{d_{A_1}} + \frac{1}{d_{A_2}} - \frac{1}{d_{A}} - \frac{1}{d_{A_1}^2} \biggr). \label{Eq:2ni4o3gsd;} 
\end{equation}
On the other hand, in the same setting, Haar random unitaries satisfy
\begin{equation}
\mathbb{E}_{U_A \sim {\sf H}}|\!|  \tr_{A_2} (U_A  \rho_{AR}  U_A^{\dagger} )- \tr_{A}J(\tr_{A_2}) \otimes  \rho_R |\!|_1 
\leq 
\frac{d_{A_1}}{\sqrt{d_{A_2}}},
\end{equation}
where the upper bound is much smaller than the right-hand side of Eq.~\eqref{Eq:2ni4o3gsd;} when $d_{A_1} = O(d_A^{1/3})$.
\end{proposition}
Proposition~\ref{Prop:example} provides an explicit example where one repetition of random $Z$- and $X$-diagonal unitaries do not generally achieve decoupling with the same error as the Haar random one.
Hence, we consider a number $\ell$ of repetitions, giving rise to a random unitary
\begin{equation}
D[\ell] := D_{\ell+1}^Z D_{\ell}^X D_{\ell}^Z \cdots D_2^X D_2^Z D_1^X D_1^Z,
\end{equation}
where $D_i^W$ are independent $W$-diagonal unitaries ($i=1,\ldots,\ell+1$, $W=X,Z$).
Note that the random unitary $D[\ell]$ starts and ends with random $Z$-diagonal unitaries, which is just for simplifying the analysis.
We denote by $\mc{D}[\ell]$ the set of unitaries on which $D[\ell]$ takes its value, and by ${\sf D}[\ell]$ the probability measure of $D[\ell]$ induced by those of $D_i^W$.

In Ref.~\cite{NHMW2015-1}, $D[\ell]$ was shown to form an approximate unitary $2$-design if $\ell$ is sufficiently large.
\begin{theorem}[Ref.~\cite{NHMW2015-1}] \label{Thm:old}
The random unitary $D[\ell]$ on $N$ qubits is a $\delta$-approximate unitary $2$-design where
\begin{equation}
\frac{2}{2^{\ell N}} \biggl[ 1 - \frac{1}{2^N-1} \biggr]\leq \delta \leq \frac{2}{2^{\ell N}} \biggl[ 1 + \frac{2}{2^N-1} \biggr].
\end{equation}
\end{theorem} 
Theorem~\ref{Thm:old} implies that $\Theta((\log 1/\delta)/N)$ repetitions are necessary and sufficient for $D[\ell]$ to achieve a $\delta$-approximate unitary $2$-design. Together with the decoupling theorem with approximate unitary $2$-designs, given in Eq.~\eqref{Eq:dec2des},
it immediately follows that decoupling at the rate $\Lambda_{\rm Haar}$ is achieved by $D[\ell]$ if $\ell \geq 4$.
However, our first main result shows that $D[\ell]$ for $\ell \geq 2$ is sufficient for achieving decoupling at rate $\Lambda_{\rm Haar}$ (see Section~\ref{Subsec:proof5} for the proof).

\begin{theorem}[Main Result 1: Decoupling with $\mathbf{D[\ell]}$] \label{Thm:decDU}
Let $\ell$ be a natural number. In the setting as Theorem~\ref{Thm:decs}, it follows that
\begin{equation}
\mathbb{E}_{U_A \sim {\sf D}[\ell]}|\!| \mc{T}_{A \rightarrow B}(U_A \rho_{AR} U_A^{\dagger}) - \tau_B \otimes  \rho_R |\!|_1 \leq  \sqrt{1 + 8 d_A^{2-\ell}} 2^{-\frac{1}{2} (H_{\rm 2}(A|R)_{\rho}+H_{\rm 2}(A|B)_{\tau} )}.
\end{equation}
In terms of the smooth conditional min-entropy, we have
\begin{equation}
\mathbb{E}_{U_A \sim {\sf D}[\ell]}|\!| \mc{T}_{A \rightarrow B}(U_A \rho_{AR} U_A^{\dagger}) - \tau_B \otimes  \rho_R |\!|_1 \leq
\sqrt{1+8d_A^{2-\ell}} 2^{-\frac{1}{2} (H^{\epsilon}_{\rm min}(A|R)_{\rho}+H^{\epsilon}_{\rm min}(A|B)_{\tau} )} + 12\epsilon. \label{Eq.epsdecDU}
\end{equation}
\end{theorem}
Noting the factor $d_A^{2-\ell}$ in the coefficient, we observe that decoupling at the rate $\Lambda_{\rm Haar}$ can be achieved by $D[\ell]$ if $\ell \geq 2$.
This result can be rephrased in terms of the approximations of $2$-designs. Because $D[2]$ is a $\Theta(d_A^{-2})$-approximate unitary $2$-design and cannot be better due to Theorem~\ref{Thm:old},
this provides the first instance of a random unitary that is less uniform than a $O(d_A^{-4})$-approximate unitary $2$-design but is able to achieve decoupling at the rate $\Lambda_{\rm Haar}$.

We also provide a probabilistic statement about decoupling with $D[\ell]$. 
\begin{theorem}[Main Result 2: Probabilistic decoupling theorem with $\mathbf{D[\ell]}$] \label{Thm:CoMtorus}
For a unitary $U_A$ drawn uniformly at random according to the probability measure ${\sf D[\ell]}$, and $\eta \geq 0$, it follows that 
\begin{equation}
{\rm Prob}_{U_A \sim {\sf D}[\ell]} \biggl[ |\!| \mc{T}_{A \rightarrow B}(U_A \rho_{AR} U_A^{\dagger}) - \tau_B \otimes  \rho_R |\!|_1
 \leq
2\Delta_{D[\ell]} + \eta  \biggr]
\geq  
1- 2 e^{- \chi_{\ell} d_A \eta^4 },
\end{equation}
where $\Delta_{D[\ell]}=\sqrt{1+8d_A^{2-\ell}} 2^{-\frac{1}{2} (H^{\epsilon}_{\rm min}(A|R)_{\rho}+H^{\epsilon}_{\rm min}(A|B)_{\tau} )} + 12\epsilon$ for $\epsilon>0$,
and $\chi_{\ell}^{-1}=2^{11}(2 \ell+1)^3 K^4 \pi^2$, and $K=d_A|\!| \rho_A |\!|_{\infty}$.
\end{theorem}
The proof is given in Section~\ref{Proof:CM}.\\

We here comment on an efficient implementation of $D[\ell]$ and the amount of randomness needed for the implementation. 
Because $D[\ell]$ consists of random $X$- and $Z$-diagonal unitaries and both cannot be exactly implemented by quantum circuits in efficient ways~\cite{NKM2014}, exact $D[\ell]$ cannot be implemented efficiently. 
It is also a problem that drawing a random phase from $[0, 2 \pi)$ uniformly at random requires an infinite number of random bits, which is surely inefficient.
To overcome these problems, we use the standard trick that 
decoupling is proved by considering first and second moments of random unitaries (see e.g. Ref.~\cite{DBWR2010}).
That is, if we use a random unitary that simulates the second order moments of $D[\ell]$, we obtain exactly the same result as Theorem~\ref{Thm:decDU}.
Simulating lower order moments of random $Z$-diagonal unitaries was thoroughly investigated in Ref.~\cite{NKM2014}, from which we obtain conclusions that the same statement as Theorem~\ref{Thm:decDU} holds not only for $D[\ell]$ but also for quantum circuits, where the number of gates is $\Theta(N_A^2)$ and the same amount of random bits is used.
This is as efficient as most other implementations of unitary $2$-designs leading to decoupling~\cite{DLT2002, DCEL2009, HL2009}, both in terms of the number of gates and the amount of randomness, but there is a better implementation, which uses $O(N_A \log^2 N_A)$ gates and $O(N_A)$ random bits~\cite{CLLW2015}. Our circuits have however a simple structure, leading to interesting insights in certain situations. For instance, it is possible to interpret the circuits as dynamics generated by many-body Hamiltonians. Then, the corresponding time-dependent Hamiltonians turn out to be certain types of quantum chaos, implying that decoupling is a typical phenomena in chaotic many-body systems. For more details, see Ref.~\cite{NHMW2015-1, NHKW2016}. 

Finally, we briefly mention implications of our results to two important applications of decoupling. For more elaborated explanations and formal statements, see Appendix~\ref{App:decApp}.

\emph{Coherent state merging}--
The task of coherent state merging~\cite{ADHW2009, DH2011} is to transfer some part $A$ of a given tripartite quantum state $\Psi_{ABR}$ from one party to another which already possess a part $B$ of the state, while leaving the reference $R$ unchanged. 
At the same time both parties would like to generate as much shared entanglement as possible.
This protocol is known as the mother of all protocols~\cite{ADHW2009} and is of significant importance. 
Our result offers the opportunity to implement the encoding for coherent state merging using $Z$- and $X$-diagonal unitaries, which achieves almost optimal rates~\cite{DH2011} when $\ell \geq 2$.

\emph{Relative thermalisation}--
One of the most fundamental questions in quantum thermodynamics is how the quantum state in a part of large system spontaneously thermalises. 
As thermal states naturally lose any correlations with other systems, the question can be extended to the idea of \emph{relative thermalisation}~\cite{dRHRW2014}.
It is known that, assuming that the time-evolution is modelled by a Haar random unitary, a system is typically decoupled from other systems and 
relative thermalisation is naturally derived~\cite{dRHRW2014}.
Since our result implies that typical dynamics in quantum chaotic systems, modelled by $D[\ell]$, achieves decoupling as strong as Haar random unitaries, we obtain an interesting observation that \emph{relative thermalisation shall be a typical phenomena in certain quantum chaotic systems}.

\section{Proofs} \label{Proofs}

Here, we provide proofs of our main results. We start with presenting the key lemmas in Subsection~\ref{SS:AL}. A proof of Theorems~\ref{Thm:decDU} and~\ref{Thm:CoMtorus} is given in Subsection~\ref{Subsec:proof5} and \ref{Proof:CM}, respectively.

\subsection{Key Lemmas} \label{SS:AL}

Here, we provide two key lemmas of the proofs of Theorems~\ref{Thm:decDU} and~\ref{Thm:CoMtorus}.
Other lemmas needed to prove Theorem~\ref{Thm:decDU} are listed in Appendix~\ref{App:AL}.

The most important lemma in the proof of Theorem~\ref{Thm:decDU} is about the CPTP map $\mc{R}$ defined by $\mc{R}= \mc{G}_{D^Z}^{(2)} \circ \mc{G}_{D^X}^{(2)} \circ \mc{G}_{D^Z}^{(2)}$, where $\mc{G}_{U}^{(2)}$ for a random unitary $U$ is given by $\mc{G}_{U}^{(2)}(X) = \mathbb{E}_{U}[U^{\otimes 2} X U^{\dagger \otimes 2}]$.
The $\ell$ repetitions of this map are known to be decomposed into a probabilistic mixture of two CPTP maps.
%The key observation to obtain our results is that the CPTP map $\mc{R}$ has a strong ability of randomisation, at least up to the second order, which is formally given in the following lemma.
 
\begin{lemma}[Lemma 2 in Ref.~\cite{NHMW2015-1}] \label{Lemma:ME}
Let $\ell$ be a natural number. Then, $\ell$ repetitions of the CPTP map $\mc{R}$, denoted by $\mc{R}^{\ell}$, is given by
$\mc{R}^{\ell} = (1-p_{\ell}) \mc{G}_{\sf H}^{(2)} + p_{\ell} \mc{C}^{(\ell)}$,
where $p_{\ell}=\frac{d^{\ell+1}+d^{\ell}-2}{d^{2\ell}(d-1)} = \Theta(d^{-\ell})$, and $\mc{C}^{(\ell)}$ is a unital CPTP map.
\end{lemma}

For the proof of the probabilistic statement of decoupling, i.e. Theorem~\ref{Thm:CoMtorus}, we use the concentration of measure on a hypertorus. Let us consider a $K$-dimensional torus $\mathbb{T}^{K}$ equipped with a normalized product measure $\mu$ and the normalized $\ell^1$-metric $d(\vec{s}, \vec{t})$, given by $\frac{1}{K}\sum_{i=1}^K |s_i - t_i|$.
Let $\vartheta(r)$ for $r>0$ be the concentration function on it defined by $\vartheta(r) = \sup\{1- \mu(A_r) : A \subset \mathbb{T}^{K}, \mu(A) \geq 1/2 \}$, where $A_r = \{ x \in \mathbb{T}^{K}: d(x, A) <r \}$. 

\begin{lemma}[Corollary of Theorem 4.4 in Ref.~\cite{L2001}] \label{Lemma:CoM}
The concentration function on $(\mathbb{T}^{K}, \mu, d)$, defined above, satisfies $\vartheta(r) \leq 2 e^{-\frac{K r^2}{8}}$.
\end{lemma}

\subsection{Proof of Theorem~\ref{Thm:decDU}} \label{Subsec:proof5}

We now prove Theorem~\ref{Thm:decDU}. 
We first use the standard technique of investigating decoupling (see e.g. Ref.~\cite{DBWR2010}),
and then combine it with the key lemma to obtain the statement.

We introduce a CP map $\mc{E}_{\bar{A} \rightarrow R}$ from $\mc{B}(\mc{H}_{\bar{A}})$ to $\mc{B}(\mc{H}_R$) satisfying $J(\mc{E}_{\bar{A} \rightarrow R})=\rho_{AR}$, where $\mc{H}_{\bar{A}} \cong \mc{H}_{A}$. 
%Note that $\mc{E}_{\bar{A} \rightarrow R}$ is not trace-preserving in general as $\rho_R$ is not necessarily the completely mixed state (see Lemma~\ref{Lemma:PropJ}).
In terms of $\mc{E}_{\bar{A} \rightarrow R}$, the resulting state after the process is given by
\begin{equation}
\mc{T}_{A \rightarrow B}  (U_A \rho_{AR} U_A^{\dagger})
=
\mc{T}_{A \rightarrow B} \otimes \mc{E}_{\bar{A} \rightarrow R}  (U_A \Phi_{A\bar{A}} U_A^{\dagger}),
\end{equation}
where $\Phi_{A\bar{A}}$ is the maximally entangled state between $A$ and $\bar{A}$. 
Denoting $\Phi_{A \bar{A}} - I_A/d_A \otimes I_{\bar{A}}/d_A$ by $\xi_{A \bar{A}}$, 
$|\!| \mc{T}_{A \rightarrow B} \otimes \mc{E}_{\bar{A} \rightarrow R} (U_A \Phi_{A\bar{A}} U_A^{\dagger}) - \tau_B \otimes  \rho_R |\!|_1$
is written as
$|\!| \mc{T}_{A \rightarrow B} \otimes \mc{E}_{\bar{A} \rightarrow E} (U_A \xi_{A \bar{A}} U_A^{\dagger}) |\!|_1$.
Following the same approach in Ref.~\cite{DBWR2010}, we obtain
\begin{equation}
\bigl( \mathbb{E}_{U_A \sim {\sf D}[\ell]}|\!| \mc{T}_{A \rightarrow B} \otimes \mc{E}_{\bar{A} \rightarrow R}(U_A \xi_{A \bar{A}} U_A^{\dagger}) |\!|_1 \bigr)^2
\leq
\tr \mathbb{E}_{U_A \sim {\sf D}[\ell]} \bigl[ (U_A \xi_{A \bar{A}} U_A^{\dagger})^{\otimes 2} \bigr] \tilde{\mc{T}}_{B \rightarrow A}^{* \otimes 2}(\mathbb{F}_{BB'}) \otimes \tilde{\mc{E}}_{R \rightarrow \bar{A}}^{* \otimes 2} (\mathbb{F}_{RR'}),
\end{equation}
where $\tilde{\mc{T}}_{A \rightarrow B}(\rho) := \sigma_B^{-1/4} \mc{T}_{A \rightarrow B}(\rho) \sigma_B^{-1/4}$, $\tilde{\mc{E}}_{\bar{A} \rightarrow R}(\rho) := \sigma_{R}^{-1/4} \mc{E}_{\bar{A} \rightarrow R}(\rho) \sigma_{R}^{-1/4}$ for any $\sigma_{W} \in \mc{S}(\mc{H}_W)$ ($W=B,R$), 
$*$ represents the adjoint in terms of the Hilbert-Schmidt inner product, and 
$\mathbb{F}_{WW'} \in \mc{B}(\mc{H}_{W}^{\otimes 2})$ $(W=B,R)$ is the swap operator defined by $\sum_{i,j} \ketbra{ij}{ji}_{WW'}$.

Now, we rewrite the average operator $\mathbb{E}_{U_A \sim {\sf D}[\ell]} \bigl[ (U_A \xi_{A \bar{A}} U_A^{\dagger})^{\otimes 2} \bigr]$ in terms of the map $\mc{R}_{AA'}$. To this end, we use a fact that applying two random $Z$-diagonal unitaries is equivalent to applying only one, which is due to the uniform distribution of the unitaries.
Hence, the random unitary $D[\ell]$ is equivalently given by
$D[\ell] =  \prod_{i=\ell}^1 \tilde{D}_i^{Z} D_i^X D_i^Z$, where all random diagonal unitaries are independent and $\prod_{i=\ell}^1$ is a product in the reverse order.
Recalling the definition of $\mc{R}_{AA'}$ and noting that, when an average is taken over all random unitaries, the independence of random unitaries enables us to take averages of each random unitary separately, we obtain 
\begin{equation}
\mathbb{E}_{U_A \sim {\sf D}[\ell]} \bigl[ (U_A \xi_{A \bar{A}} U_A^{\dagger})^{\otimes 2} \bigr]
=
\mc{R}_{AA'}^{\ell}(\xi_{A \bar{A}} \otimes \xi_{A' \bar{A'}}),
\end{equation}
which leads to an upper bound in terms of $\mc{R}_{AA'}$:
\begin{align}
\bigl( \mathbb{E}_{U_A \sim {\sf D}[\ell]}|\!| \mc{T}_{A \rightarrow B} \otimes \mc{E}_{\bar{A} \rightarrow E}(U_A \xi_{A \bar{A}} U_A^{\dagger}) |\!|_1 \bigr)^2
& \leq
\tr  \mc{R}_{AA'}^{\ell}(\xi_{A \bar{A}} \otimes \xi_{A' \bar{A'}})\tilde{\mc{T}}_{B \rightarrow A}^{* \otimes 2}(\mathbb{F}_{BB'}) \otimes \tilde{\mc{E}}_{R \rightarrow \bar{A}}^{* \otimes 2} (\mathbb{F}_{RR'}). \label{Eq:erk}
\end{align}

Due to Lemma~\ref{Lemma:ME}, the map $\mc{R}_{AA'}^{\ell}$ can be further decomposed into the map corresponding to a Haar random unitary and the unital CPTP map $\mc{C}^{(\ell)}$, leading to
\begin{multline}
\bigl( \mathbb{E}_{U_A \sim {\sf D}[\ell]}|\!| \mc{T}_{A \rightarrow B} \otimes \mc{E}_{\bar{A} \rightarrow E}(U_A \xi_{A \bar{A}} U_A^{\dagger}) |\!|_1 \bigr)^2
\leq
(1-p_{\ell}) \tr   \mc{G}_{{\sf H }, AA'}^{(2)} (\xi_{A \bar{A}} \otimes \xi_{A' \bar{A'}}) \tilde{\mc{T}}_{B \rightarrow A}^{* \otimes 2}(\mathbb{F}_{BB'}) \otimes \tilde{\mc{E}}_{R \rightarrow \bar{A}}^{* \otimes 2} (\mathbb{F}_{RR'})\\
+
p_{\ell} \tr \mc{C}_{AA'}^{(\ell)}(\xi_{A \bar{A}} \otimes \xi_{A' \bar{A'}}) \tilde{\mc{T}}_{B \rightarrow A}^{* \otimes 2}(\mathbb{F}_{BB'}) \otimes \tilde{\mc{E}}_{R \rightarrow \bar{A}}^{* \otimes 2} (\mathbb{F}_{RR'}) \label{Eq:334t,}
\end{multline}
The first term of the right-hand side in Eq.~\eqref{Eq:334t,} is exactly the same as the term investigated in Ref.~\cite{DBWR2010} and is bounded from above as
\begin{equation}
\tr   \mc{G}_{{\sf H }, AA'}^{(2)} (\xi_{A \bar{A}} \otimes \xi_{A' \bar{A'}}) \tilde{\mc{T}}_{B \rightarrow A}^{* \otimes 2}(\mathbb{F}_{BB'}) \otimes \tilde{\mc{E}}_{R \rightarrow \bar{A}}^{* \otimes 2} (\mathbb{F}_{RR'})
\leq 
\tr \tilde{\tau}_{AB}^2 \tr \tilde{\rho}_{AR}^2,
\end{equation}
where $\tilde{\tau}_{AB}=J(\tilde{\mc{T}}_{A \rightarrow B})$ and $\tilde{\rho}_{AB}=J(\tilde{\mc{E}}_{\bar{A} \rightarrow R})$.
For the second term of the right-hand side in Eq.~\eqref{Eq:334t,}, we substitute $\xi_{A \bar{A}} = \Phi_{A \bar{A}} - I_A/d_A \otimes I_{\bar{A}}/d_A$ and obtain
\begin{equation}
\mc{C}^{(\ell)}_{AA'}(\xi_{A \bar{A}} \otimes \xi_{A' \bar{A'}})
=
\mc{C}^{(\ell)}_{AA'}(\Phi_{A \bar{A}} \otimes \Phi_{A' \bar{A'}}) - I_{A \bar{A} A' \bar{A'}}/d_A^4, \label{Eq:50}
\end{equation}
where we have used a fact that $\mc{C}^{(\ell)}$ is unital. Then, the second term of the right-hand side in Eq.~\eqref{Eq:334t,} is bounded from above by
\begin{align}
&\tr \mc{C}_{AA'}^{(\ell)}(\xi_{A \bar{A}} \otimes \xi_{A' \bar{A'}}) \tilde{\mc{T}}_{B \rightarrow A}^{* \otimes 2}(\mathbb{F}_{BB'}) \otimes \tilde{\mc{E}}_{R \rightarrow \bar{A}}^{* \otimes 2} (\mathbb{F}_{RR'}) \\
&\leq
|\!| \mc{C}_{AA'}^{(\ell)}(\xi_{A \bar{A}} \otimes \xi_{A' \bar{A'}}) \tilde{\mc{T}}_{B \rightarrow A}^{* \otimes 2}(\mathbb{F}_{BB'}) \otimes \tilde{\mc{E}}_{R \rightarrow \bar{A}}^{* \otimes 2} (\mathbb{F}_{RR'}) |\!|_1\\
&\leq
|\!| \mc{C}_{AA'}^{(\ell)}(\xi_{A \bar{A}} \otimes \xi_{A' \bar{A'}}) |\!|_1
|\!| \tilde{\mc{T}}_{B \rightarrow A}^{* \otimes 2}(\mathbb{F}_{BB'})|\!|_{\infty}|\!| \tilde{\mc{E}}_{R \rightarrow \bar{A}}^{* \otimes 2} (\mathbb{F}_{RR'}) |\!|_{\infty}\\
&\leq
d_A^2|\!| \mc{C}^{(\ell)}_{AA'}(\Phi_{A \bar{A}} \otimes \Phi_{A' \bar{A'}}) - I_{A \bar{A} A' \bar{A'}}/d_A^4 |\!|_1
\tr \tilde{\tau}_{AB}^2 \tr \tilde{\rho}_{AR}^2\\
&\leq
2d_A^2 \tr \tilde{\tau}_{AB}^2 \tr \tilde{\rho}_{AR}^2,
\end{align}
where the second inequality is obtained by the H\"{o}lder's inequality, the third one by Lemma~\ref{Lemma:} in Appendix~\ref{App:AL}, and the last one by the triangle inequality and the fact that $\mc{C}^{(\ell)}_{AA'}$ is a CPTP map.

Together with these upper bounds, we obtain from Eq.~\eqref{Eq:334t,} that
\begin{align}
\bigl( \mathbb{E}_{U_A \sim {\sf D}[\ell]}|\!| \mc{T}_{A \rightarrow B} \otimes \mc{E}_{\bar{A} \rightarrow E}(U_A \xi_{A \bar{A}} U_A^{\dagger}) |\!|_1 \bigr)^2
&\leq
\bigl(1 + (2d_A^2-1)p_{\ell} \bigr) \tr \tilde{\tau}_{AB}^2 \tr \tilde{\rho}_{AR}^2\\
&\leq
\bigl(1 + 8 d_A^{2-\ell} \bigr) \tr \tilde{\tau}_{AB}^2 \tr \tilde{\rho}_{AR}^2
\end{align}
Recalling $|\!| \mc{T}_{A \rightarrow B} \otimes \mc{E}_{\bar{A} \rightarrow E} (U_A \xi_{A \bar{A}} U_A^{\dagger}) |\!|_1=|\!| \mc{T}_{A \rightarrow B}(U_A \rho_{AR} U_A^{\dagger}) - \tau_B \otimes  \rho_R |\!|_1$
and choosing $\sigma_X$ ($X=B,R$) such that $\tr \tilde{\rho}_{AR}^2=2^{-H_2(A|R)_{\rho}}$ and $\tr \tilde{\tau}_{AB}^2 = 2^{- H_2 (A|B)_{\tau}}$,
we obtain the desired result:
\begin{equation}
\mathbb{E}_{U_A \sim {\sf D}[\ell]} |\!| \mc{T}_{A \rightarrow B}(U_A \rho_{AR} U_A^{\dagger}) - \tau_B \otimes  \rho_R |\!|_1
\leq \sqrt{1 +8 d_A^{2-\ell}} 2^{- \frac{1}{2}( H_{\rm 2} (A|B)_{\tau} +H_{\rm 2}(A|R)_{\rho})}.
\end{equation}

The statement in terms of the smooth conditional min-entropies is obtained following the same approach in Ref.~\cite{DBWR2010}. $\hfill \blacksquare$

\subsection{Proof of Theorem~\ref{Thm:CoMtorus}} \label{Proof:CM}

Let $\mc{D}[\ell]$ be a set of unitaries $\{Z_{\ell+1} \prod_{i=\ell}^{1} X_i Z_i  :W_i \in \mc{D}_{W}(d), W \in \{X,Z\}, i=1,\cdots, \ell+1 \}$ on which the random unitary $D[\ell]$ takes the value. To prove Theorem~\ref{Thm:CoMtorus}, we start with a simple observation that
$\mc{D}[\ell]$ is isomorphic to a $(2 \ell +1) d_A$-dimensional torus $\mathbb{T}^{(2 \ell +1) d_A}$.
We prove Theorem~\ref{Thm:CoMtorus} based on the concentration of measure on $\mathbb{T}^{(2 \ell +1) d_A}$ equipped with the normalized product measure $\mu$ and the normalized $\ell^1$-metric given by $d(s, t) = \frac{1}{2\pi (2 \ell +1)d_A} \sum_{i=1}^{(2 \ell +1) d_A}|s_i - t_i|$ for $s,t \in \mathbb{T}^{(2 \ell +1) d_A}$.

We again use a CP map $\mathcal{E}_{\bar{A} \rightarrow R}$, defined by $J(\mathcal{E}_{\bar{A} \rightarrow R})= \rho_{AR}$.
Let $\Delta(U_A)$ be $|\!|\mc{T}_{A \rightarrow B} \otimes \mc{E}_{\bar{A} \rightarrow R} (U_A \Phi_{A\bar{A}} U_A^{\dagger}) - \tau_B \otimes  \rho_R |\!|_1$ for $U_A \in \mc{D}[\ell]$ and $\Delta_M$ be a median of $\Delta(U_A)$. 
We denote by $\phi_U$ the point on $\mathbb{T}^{(2 \ell +1) d_A}$ corresponding to $U_A$ and define
$S$ by $\{\phi_U \in \mathbb{T}^{(2 \ell +1) d_A} : \Delta(U_A) \leq \Delta_M \}$.
We also denote by $S_r$ an $r$-neighbourhood of $S$ given by $S_r= \{ \phi \in \mathbb{T}^{(2 \ell +1) d_A}: d(\phi, S) < r \}$.
As $\Delta_M$ is the median, $S$ satisfies $\mu(S) \geq 1/2$, leading to $1 - \mu(S_r) \leq \vartheta(r)$, where $\vartheta(\eta)$ is the concentration function on $(\mathbb{T}^{(2 \ell +1) d_A}, \mu, d)$.
As $\mu$ is a normalized measure, this is equivalent to
\begin{equation}
\mu \bigl( \{ \phi \in \mathbb{T}^{(2 \ell +1) d_A} :  d(\phi, S) \geq r      \}  \bigr) \leq \vartheta(r).
\end{equation}
Using Lemma~\ref{Lemma:CoM}, we have
\begin{equation}
\mu \bigl( \{ \phi \in \mathbb{T}^{(2 \ell +1) d_A} :  d(\phi, S) \geq r      \}  \bigr) \leq 2 e^{-\frac{(2\ell +1)d_A r^2}{8}}. \label{Eq:CoMtorus}
\end{equation}
We relate this inequality with the probability on $\mc{D}[\ell]$ according to ${\sf D}[\ell]$, that $\Delta(U_A)$ deviates from its median $\Delta_M$ by more than $\eta$.
To do so, we calculate an upper bound of $|\Delta(U_A) - \Delta(V_A)|$, where $U_A, V_A \in \mc{D}[\ell]$, in terms of a distance $d(\phi_U, \phi_V)$ on $\mathbb{T}^{(2 \ell +1) d_A}$.

Defining $K$ by $\max \{|\!| \mc{T}_{A \rightarrow B} \otimes \mc{E}_{\bar{A} \rightarrow R}(X)|\!|_1 : X \in \mc{B}(\mc{H}_A^{\otimes 2}), |\!| X |\!|_1 \leq 1 \}$, we have
\begin{align}
|\Delta(U_A) - \Delta(V_A)| 
&\leq 
|\!| \mc{T}_{A \rightarrow B} \otimes \mc{E}_{\bar{A} \rightarrow R}(U_A \Phi_{A \bar{A}} U_A^{\dagger} -V_A \Phi_{A \bar{A}} V_A^{\dagger}) |\!|_1 \\
&\leq 
K |\!| U_A \Phi_{A \bar{A}} U_A^{\dagger} -V_A \Phi_{A \bar{A}} V_A^{\dagger} |\!|_1 \\
&\leq 
\sqrt{2} K |\!| U_A \Phi_{A \bar{A}} U_A^{\dagger} -V_A \Phi_{A \bar{A}} V_A^{\dagger} |\!|_2 \\
&\leq 
2 K |( U_A  -V_A )\ket{\Phi_{A \bar{A}}} |, \label{Eq:70}
\end{align}
where we used the triangular inequality in the first line, a fact that the rank of $U_A \Phi_{A \bar{A}} U_A^{\dagger} -V_A \Phi_{A \bar{A}} V_A^{\dagger}$ is at most two and $|\!|A |\!|_1 \leq \sqrt{ {\rm rank} A}|\!|A |\!|_2$ in the third line, and $|\!| \ketbra{\phi}{\phi} - \ketbra{\psi}{\psi} |\!|_2 \leq \sqrt{2}| \ket{\phi} - \ket{\psi} | $ in the last line.

As $U_A, V_A \in \mc{D}[\ell]$, we denote them by $Z^{(\ell+1)}_A \prod_{m=\ell}^{1} X^{(m)}_A Z^{(m)}_A$ and $\tilde{Z}^{(\ell+1)}_A \prod_{m=\ell}^{1} \tilde{X}^{(m)}_A \tilde{Z}^{(m)}_A$, respectively, where $W^{(m)}_A, \tilde{W}^{(m)}_A \in \mc{D}_{W, {\rm diag}}$ ($W \in \{X,Z\}$ and $m= \{1,\cdots, \ell+1 \}$). 
Using this notation and using Lemma~\ref{Lemma:chainrule}, we have
\begin{align}
|\Delta(U_A) - \Delta(V_A)| 
&\leq 
2 K \biggl( \sum_{m=1}^{\ell+1}|(  Z^{(m)}_A -  \tilde{Z}^{(m)}_A )\ket{\Phi_{A \bar{A}}} |
+ \sum_{m=1}^{\ell}|(  X^{(m)}_A -  \tilde{X}^{(m)}_A )\ket{\Phi_{A \bar{A}}} | \biggr). \label{Eq:65}
\end{align} 
Denoting $Z^{(m)}_A = {\rm diag}_Z(e^{i \zeta^{(m)}_1}, \cdots, e^{i \zeta^{(m)}_{d_A}})$ and $\tilde{Z}^{(m)}_A = {\rm diag}_Z(e^{i \tilde{\zeta}^{(m)}_1}, \cdots, e^{i \tilde{\zeta}^{(m)}_{d_A}})$, we have
\begin{align}
|(  Z^{(m)}_A -  \tilde{Z}^{(m)}_A )\ket{\Phi_{A \bar{A}}} |^2
&=
\frac{1}{d_A}  \sum_{k=1}^{d_A} |e^{i \zeta^{(m)}_k}   - e^{i \tilde{\zeta}^{(m)}_k}|^2\\
&\leq 
\frac{2}{d_A}  \sum_{k=1}^{d_A} |e^{i \zeta^{(m)}_k}   - e^{i \tilde{\zeta}^{(m)}_k}|\\
&\leq 
\frac{2}{d_A}  \sum_{k=1}^{d_A} |\zeta^{(m)}_k   - \tilde{\zeta}^{(m)}_k|,
\end{align}
for any $m$,
where 
we used facts that $x^2 \leq x$ for $x \in [-1,1]$ in the second line and that $|e^{i x} - e^{i y}| \leq |x-y|$ in the last line. Thus, we have
\begin{align}
|(  Z^{(m)}_A -  \tilde{Z}^{(m)}_A )\ket{\Phi_{A \bar{A}}} |
&\leq
\biggl( \frac{2}{d_A}  \sum_{k=1}^{d_A} | \zeta^{(m)}_k - \tilde{\zeta}^{(m)}_k  | \biggr)^{1/2}. \label{Eq:ZZZZZZZZZZZZzzzz}
\end{align}
Similarly, we have 
\begin{align}
|(  X^{(m)}_A -  \tilde{X}^{(m)}_A )\ket{\Phi_{A \bar{A}}} |
&\leq
\biggl( \frac{2}{d_A}  \sum_{k=1}^{d_A} | \chi^{(m)}_k - \tilde{\chi}^{(m)}_k  |\biggr)^{1/2}, \label{Eq:XXXXXXXXXXXXXXX}
\end{align}
for $X^{(m)}_A = {\rm diag}_X(e^{i \chi^{(m)}_1}, \cdots, e^{i \chi^{(m)}_{d_A}})$ and $\tilde{X}^{(m)}_A = {\rm diag}_X(e^{i \tilde{\chi}^{(m)}_1}, \cdots, e^{i \tilde{\chi}^{(m)}_{d_A}})$.
Substituting Eqs.~\eqref{Eq:ZZZZZZZZZZZZzzzz} and~\eqref{Eq:XXXXXXXXXXXXXXX} into Eq.~\eqref{Eq:65} and using the convexity of $\sqrt{x}$, $|\Delta(U_A) - \Delta(V_A)| $ is bounded from above as follows;
\begin{align}
|\Delta(U_A) - \Delta(V_A)|  
&\leq 
2K \sqrt{\frac{2(2\ell+1)}{d_A}} \biggl( \sum_{m=1}^{\ell+1} \sum_{k=1}^{d_A} | \zeta^{(m)}_k - \tilde{\zeta}^{(m)}_k  |  +  \sum_{m=1}^{\ell} \sum_{k=1}^{d_A} | \chi^{(m)}_k - \tilde{\chi}^{(m)}_k  |\biggr)^{1/2}\\
&\leq 
4 (2 \ell +1) K \sqrt{\pi} \sqrt{d(\phi_U, \phi_V)}. \label{Eq:71}
\end{align}

The equation~\eqref{Eq:71} implies that, to change the value of $\Delta(U_A)$ by more than $\eta$ from its median $\Delta_M$, it is necessary to change the value of $d(\phi_U, \phi_V)$ by at least $\frac{\eta^2}{2^4(2 \ell +1)^2 K^2 \pi}$.
Thus, we obtain from Eq.~\eqref{Eq:CoMtorus} that
\begin{align}
{\rm Prob}_{U_A \sim {\sf D}[\ell]}[ \Delta(U_A) - \Delta_M \geq \eta ]
&\leq 
\mu \bigl( \{ \phi \in \mathbb{T}^{(2 \ell +1) d_A} : d(\phi, S) \geq    \frac{\eta^2}{2^4(2 \ell +1)^2 K^2 \pi}   \}  \bigr) \\
&\leq 
2 \exp \bigl[ -\frac{d_A \eta^4}{2^{11} (2 \ell +1)^3 K^4 \pi^2 }    \bigr].
\end{align}
The median is replaced with the expectation value using Markov's inequality $\Delta_M \leq 2 \mathbb{E}_{U_A \sim {\sf D}[\ell]} \Delta(U_A)$, and we can use an upper bound of the expectation value given in Eq.~\eqref{Eq.epsdecDU} in terms of the conditional smooth entropies.

Finally, we investigate $K$. As $K:=\max \{|\!| \mc{T}_{A \rightarrow B} \otimes \mc{E}_{\bar{A} \rightarrow R}(X)|\!|_1 : X \in \mc{B}(\mc{H}_A \otimes \mc{H}_{\bar{A}}), |\!| X |\!|_1 \leq 1 \}$, it is bounded from above by $|\!| \mc{T}_{A \rightarrow B} \otimes \mc{E}_{\bar{A} \rightarrow R}|\!|_{1 \rightarrow 1}=\max \{|\!| \mc{T}_{A \rightarrow B} \otimes \mc{E}_{\bar{A} \rightarrow R}(\sigma_{A\bar{A}})|\!|_1 : \sigma_{A\bar{A}} \in \mc{S}(\mc{H}_A \otimes \mc{H}_{\bar{A}}) \}$.
Using the Choi-Jamio\l kowski representation $\tau_{AB} = J(\mc{T}_{A \rightarrow B})$ and $\rho_{AB} = J(\mc{E}_{\bar{A} \rightarrow R})$ (see Lemma.~\ref{Lemma:CJ}), we obtain
\begin{align}
|\!| \mc{T}_{A \rightarrow B} \otimes \mc{E}_{\bar{A} \rightarrow R}(\sigma_{A\bar{A}})|\!|_1
&=
d_A^2 |\!|\tr_{A \bar{A}} (\sigma_{A \bar{A}}^T \otimes I_{BR})(\tau_{AB} \otimes \rho_{\bar{A}R})|\!|_1\\
&= d_A^2 \tr (\sigma_{A \bar{A}}^T \otimes I_{BR})(\tau_{AB} \otimes \rho_{\bar{A}R})\\
&= d_A^2 \tr \sigma_{A \bar{A}}^T (\tau_{A} \otimes \rho_{\bar{A}})\\
&= d_A \tr \sigma_{\bar{A}}^T  \rho_{\bar{A}},
\end{align}
where we used that $\tau_A = I_A/d_A$, as $\mc{T}_{A \rightarrow B}$ is a TP map, in the last line.
This implies $K \leq d_A |\!| \rho_A |\!|_{\infty}$.
 $\hfill \blacksquare$

\section{Conclusion} \label{Conclusion}

We have investigated decoupling with $D[\ell]$, a random unitary obtained by alternately applying independent random $Z$- and $X$-diagonal unitaries $\ell$ times.
We have shown that, if $\ell \geq 2$, $D[\ell]$ achieves decoupling at the same rate as the Haar random one.
As $D[2]$ forms a $\Theta(d_A^{-2})$-approximate unitary $2$-design and cannot be better, our result implies that decoupling at the same rate as the Haar random one is achievable with rather imprecise approximations of unitary $2$-designs.
Although this construction itself is not efficient, we have explained that the same results can be obtained by efficient quantum circuits and by the dynamics generated by time-dependent chaotic Hamiltonians.
We have also provided a concentration of measure phenomena for the decoupling with $D[\ell]$ and have briefly discussed implications of our results to information theoretic protocols and relative thermalisation.

A natural open question is if worse approximations of unitary $2$-designs can achieve the same rate of decoupling.
Since we have presented that $\Theta(d_A^{-2})$-approximations can achieve decoupling, what is not clear is whether $\Omega(d_A^{-1})$-approximations could achieve decoupling. This question can be rephrased in terms of $D[\ell]$. 
Since we have shown that $D^X D^Z$ does not achieve decoupling in general (Proposition~\ref{Prop:example}) and that $D[2]=\tilde{\tilde{D}}^Z \tilde{D}^X \tilde{D}^Z D^X D^Z$ achieves decoupling well (Theorem~\ref{Thm:decDU}), it is interesting to investigate decoupling with $D[1] = \tilde{D}^Z D^X D^Z$, which is a $\Theta(d_A^{-1})$-approximate unitary $2$-design.

Another natural and theoretically interesting question is whether random unitaries strictly less uniform than unitary $2$-designs could achieve the Haar decoupling rate. This is important to figure out what properties of Haar random unitaries enable us to achieve decoupling.
One possible way to address this question is to use random quantum circuits and to analyse the relation between the length of the circuits and the decoupling rate achievable with the circuit, which is in the same spirit as Ref.~\cite{BF2013}. 
Another approach is to address the question directly in terms of unitary designs.
To this end, it may be significant to generalise the idea of unitary $t$-designs, defined so far only for a natural number $t$, to ``finer'' versions by introducing smooth parameters quantifying the uniformity of random unitaries, such as \emph{unitary $t$-designs for positive real numbers $t$}. It is not clear if such unitary designs can be defined in the way consistent with the standard definitions of designs. Nevertheless, we believe that addressing this question will lead to new developments of the theory of random unitaries.

Finally, in Ref.~\cite{NHMW2015-1} and in this paper, we have shown that diagonal unitaries, which can implement a unitary $2$-design and achieve decoupling, are potentially useful in quantum information science. 
Our constructions use quantum circuits with a constant number of commuting parts, which may lead to a simplification of the experimental realisation to some extent. However, in many tasks, a unitary $2$-design or decoupling is just one primitive of the whole process, and non-commuting gates are required in the rest of the protocols. It is hence interesting and practically useful to invent the method of re-organising the quantum circuit so that it consists of the minimal number of commuting parts.

\section{Acknowledgements}
The authors are grateful to R. F. Werner and O. Fawzi for interesting
and fruitful discussions. 
YN is a JSPS Research Fellow and is supported in part by  JSPS Postdoctoral Fellowships for Research Abroad, and by JSPS KAKENHI Grant Number 272650.
CH and CM acknowledge support from the EU grants SIQS and QFTCMPS and by the cluster of excellence EXC 201 Quantum Engineering and Space-Time Research.
AW is supported by the European Commission (STREP ``RAQUEL''), the European Research Council (Advanced Grant ``IRQUAT''), the Spanish MINECO, project FIS2008-01236, with the support of FEDER funds.
CH and AW are also supported by the Generalitat de Catalunya, CIRIT project no. 2014 SGR 966, as well as the Spanish MINECO, projects FIS2013-40627-P and FIS2016-80681-P (AEI/FEDER, UE) and CH by FPI Grant No. BES-2014-068888.

\appendix 

\section{Notation in detail} \label{App:ND}

Let $f(n)$ and $g(n)$ be functions on $\mathbb{R}^+$.
We say $f(n) = O(g(n))$ if there exist $c, n_0 >0$ such that $f(n) \leq c g(n)$ for all $n \geq n_0$. 
When there exist $c, n_0 >0$ such that $f(n) \geq c g(n)$ for all $n \geq n_0$, we say $f(n) = \Omega(g(n))$.
If $f(n) = O(g(n))$ and $f(n) = \Omega(g(n))$, we denote it by $f(n) = \Theta(g(n))$.

The 1-norm, the 2-norm and the operator norm of $X \in \mc{B}(\mc{H})$ are given by $|\!| X |\!|_1 := \tr \sqrt{X^{\dagger} X}$, $|\!| X |\!|_2 := \sqrt{ \tr X^{\dagger} X}$, and $|\!| X |\!|_{\infty} := \max_i x_i$, where $\{x_i\}$ are the singular values of $X$, respectively.
They satisfy $|\!| X |\!|_1 \geq |\!| X |\!|_2 \geq |\!| X |\!|_{\infty}$ for any $X \in \mc{B}(\mc{H})$.
The 1-norm is of particular importance in quantum information science because it provides the optimal success probability $(1+|\!| \rho - \sigma |\!|_1/2)/2$ when we would like to distinguish two quantum states $\rho$ and $\sigma$.
For a superoperator $\mc{C} : \mc{B}(\mc{H}) \rightarrow \mc{B}(\mc{H})$, the diamond norm is defined by
\begin{equation}
|\!| \mc{C} |\!|_{\diamond} := \sup_k  
\sup_{X \neq 0 }\frac{|\!|\mc{C} \otimes {\rm id}_k(X)|\!|_1}{|\!|X|\!|_1},
\end{equation}
where ${\rm id_k}$ is the identity map acting on a Hilbert space of dimension $k$. 
Similarly to the $1$-norm for quantum states, the diamond norm provides the optimal success probability to distinguish two quantum channels when we are allowed to use auxiliary systems. 

\section{Additional Lemmas} \label{App:AL}

Here, we list lemmas. The proofs are given in the referred papers or right after the lemma.

\begin{lemma}[{\bf Choi-Jamio\l kowski representation~\cite{J1972,C1975}}] \label{Lemma:CJ}
The Choi-Jamio\l kowski map $J$ is an isomorphism between a set of linear maps from $\mc{B}(\mc{H}_A)$ to $\mc{B}(\mc{H}_B)$ and $\mc{B}(\mc{H}_{AB})$.
The map $J$ is explicitly given by 
\begin{equation}
J(\mc{T}_{A \rightarrow B}) = ({\rm id}_A \otimes \mc{T}_{A' \rightarrow B})(\ketbra{\Phi}{\Phi}_{AA'}),
\end{equation}
where ${\rm id}_A$ is the identity map on $\mc{H}_A$, $\ket{\Phi}_{AA'}=\frac{1}{\sqrt{d_A}}\sum_i \ket{ii}_{AA'}$ is the maximally entangled state in $\mc{H}_{A A'}$, and $\mc{H}_{A'} \cong \mc{H}_{A}$.
The inverse map $J^{-1}$ takes an operator $X_{AB} \in \mc{B}(\mc{H}_{AB})$ to the linear map $\mc{T}_{A \rightarrow B}$ given by
\begin{equation}
\mc{T}_{A\rightarrow B}(Y_A) = d_A \tr_A (Y_A^T \otimes I_B)X_{AB}
\end{equation}
for any $Y_A \in \mc{B}(\mc{H}_A)$, where $T$ represents a transpose with respect to the Schmidt basis used in $\ket{\Phi}_{AA'}$.
\end{lemma}

%\begin{lemma}[Lemma 3.7 in Ref.~\cite{RennerThesis}] \label{Lemma:TraceTwo}
%Let $H$ be an Hermitian operator on $\mathcal{H}$, then, for any $\sigma \in  \mc{P}(\mathcal{H})$, it follows that $|\!| H |\!|_1 \leq \sqrt{ \tr{\sigma} } |\!|\sigma^{-1/4} H \sigma^{-1/4} |\!|_2$.
%\end{lemma}

\begin{lemma}[Swap trick (see e.g. Ref.~\cite{L2010})] \label{Lemma:SwapTrick}
Let $X, Y$ in $\mc{B}(\mc{H})$. Then, $\tr (X \otimes Y) \mathbb{F} = \tr XY$.
\end{lemma}

\begin{lemma}[Ref.~\cite{SDTR2013}] \label{Lemma:}
For a linear map $\mc{C}_{A \rightarrow B}$ from $\mc{B}(\mc{H}_A)$ to $\mc{B}(\mc{H}_B)$, its adjoint map $\mc{C}_{B \rightarrow A}^*$ in terms of the Hilbert-Schmidt inner product, defined by $\tr \mc{C}_{A \rightarrow B}(X_A) Y_B = \tr X_A \mc{C}_{B \rightarrow A}^*(Y_B)$, satisfies
$|\!| \mc{C}_{B \rightarrow A}^{* \otimes 2}(\mathbb{F}_{BB'}) |\!|_{\infty} \leq d_A \tr J(\mc{C}_{A \rightarrow B})^2$.
\end{lemma}

Finally, we use a simple inequality for the vector norm. 

\begin{lemma} \label{Lemma:chainrule}
Let $\ket{\Phi_{AA'}}$ be a vector corresponding to the maximally entangled state in $\mc{H}_{AA'}$. For unitaries $U_A, U'_A, V_A, V'_A$ acting on $\mc{H}_A$, it holds that
\begin{equation}
|(U_A U'_A - V_A V'_A) \ket{\Phi_{AA'}} |  \leq |(U_A - V_A) \ket{\Phi_{AA'}} | +|(U'_A - V'_A) \ket{\Phi_{AA'}} |.
\end{equation}
\end{lemma}

\begin{proof}
Using the triangular inequality, we have
\begin{equation}
|(U_A U'_A - V_A V'_A) \ket{\Phi_{AA'}} |  \leq |(U_A - V_A) U'_A \ket{\Phi_{AA'}} | + |V_A(U'_A - V'_A) \ket{\Phi_{AA'}} |.
\end{equation}
As the norm is invariant under a unitary transformation, the second term on the right-hand side is simply equal to $|(U'_A - V'_A) \ket{\Phi_{AA'}} |$.
Using the property of the maximally entangled state such that $W_A \ket{\Phi_{AA'}} = W_{A'}^T \ket{\Phi_{AA'}}$, where $T$ is the transpose in the Schmidt basis of $\ket{\Phi_{AA'}}$ in $\mc{H}_{A'}$,
the first term on the right-hand side is given by
$|(U_A - V_A) (U'_{A'})^T \ket{\Phi_{AA'}} |$, which is also equal to
$| (U'_{A'})^T (U_A - V_A) \ket{\Phi_{AA'}} | = | (U_A - V_A) \ket{\Phi_{AA'}} |$. 
\end{proof}

\section{Proof of Proposition~\ref{Prop:example}} \label{Ex:Proof}

As $|\!| X |\!|_1 \geq |\!| X |\!|_2$ for any $X \in \mc{B}(\mc{H})$, we obtain
\begin{align}
\mathbb{E}_{U_A \sim {\sf D}_X \times  {\sf D}_Z}|\!| \mc{T}_{A \rightarrow A_1} (U_A  \rho_{AR}  U_A^{\dagger} )- \tau_{A_1} \otimes  \rho_R |\!|_1
&\geq 
\mathbb{E}|\!| \mc{T}_{A \rightarrow A_1} (U_A  \rho_{AR}  U_A^{\dagger} )- \tau_{A_1} \otimes  \rho_R |\!|_2\\
&=\mathbb{E} \sqrt{ \tr \bigl( \mc{T}_{A \rightarrow A_1} (U_A  \rho_{AR}  U_A^{\dagger} )- \tau_{A_1} \otimes  \rho_R \bigr)^2}\\
&\geq 
\frac{1}{\sqrt{2}}
\mathbb{E}  \tr \bigl( \mc{T}_{A \rightarrow A_1} (U_A  \rho_{AR}  U_A^{\dagger} )- \tau_{A_1} \otimes  \rho_R \bigr)^2
\end{align}
where we used an inequality $\sqrt{x} \leq x$ for $0 \leq x \leq 1$ and a fact that $\tr (\rho - \sigma)^2 \leq 2$ for any $\rho, \sigma \in \mc{S}(\mc{H})$.
Substituting $\tau_{A_1} = \rho_R=I_{A_1}/d_{A_1}$ and using a swap trick (Lemma~\ref{Lemma:SwapTrick}) and $\mathbb{E}_{U_A \sim {\sf D}_X \times {\sf D}_Z} [U_A  \rho_{AR}  U_A^{\dagger}] = I_{A}/d_A \otimes \rho_R$, we obtain
\begin{align}
\mathbb{E}_{U_A \sim {\sf D}_X \times {\sf D}_Z}|\!| \mc{T}_{A \rightarrow A_1} (U_A  \rho_{AR}  U_A^{\dagger} )- \tau_{A_2} \otimes  \rho_R |\!|_1
&\geq
\frac{1}{\sqrt{2}}\biggl(
\mathbb{E} \bigl[\tr  \bigl( \mc{T}_{A \rightarrow A_1} (U_A  \rho_{AR}  U_A^{\dagger} ) \bigr)^2 \bigr]	 - 1/d_{A_1}^2 \biggr). \label{Eq:last}
\end{align}
Using a swap trick given in Lemma~\ref{Lemma:SwapTrick} and linearity of $\mc{T}_{A \rightarrow A_1}$, we have
\begin{align}
\mathbb{E}_{U_A \sim {\sf D}_X \times {\sf D}_Z} [\tr \bigl( \mc{T}_{A \rightarrow A_1} (U_A  \rho_{AR}  U_A^{\dagger} ) \bigr)^2]
%&=\mathbb{E}\tr \bigl( \mc{T}_{A \rightarrow A_1} (U_A  \rho_{AR}  U_A^{\dagger} ) \bigr)^{\otimes 2} \mathbb{F}_{A_1 A_1'} \otimes \mathbb{F}_{RR'}\\
&=\tr \mc{T}_{A \rightarrow A_1}^{\otimes 2} \bigl( \mathbb{E} \bigl[ (U_A  \rho_{AR}  U_A^{\dagger} )^{\otimes 2} \bigr]  \bigr) \mathbb{F}_{A_1 A_1'} \otimes \mathbb{F}_{RR'}. \label{Eq:12345}
\end{align}

We then directly compute the operator $\mathbb{E} (U_A  \rho_{AR}  U_A^{\dagger} )^{\otimes 2} $ by substituting $\rho_{AR}=\Phi_{A_1 R} \otimes \ketbra{0}{0}_{A_2}$ and using relations such as 
$\mathbb{E} \bigl[ (D^Z_A  O D_A^{Z \dagger} )^{\otimes 2} \bigr]=\sum_{i,j} (O^{ij}_{ij} \ketbra{ij}{ij} + O^{ij}_{ji} \ketbra{ij}{ji} - O^{ii}_{ii}\ketbra{ii}{ii})$ for any $O=\sum_{i,j,k,l}\ketbra{ij}{kl} \in \mc{B}(\mc{H}^{\otimes 2})$, where $\{\ket{i}\}$ is the Pauli-$Z$ basis.
Noting that $U_A = D^X D^Z$ and $\mc{T}_{A \rightarrow A_1} = \tr_{A_2}$, we obtain after some calculation that
\begin{multline}
\mc{T}_{A \rightarrow A_1}^{\otimes 2}  \bigl(\mathbb{E}_{U_A \sim {\sf D}_X \otimes {\sf D}_Z} \bigl[ (U_A  \rho_{AR}  U_A^{\dagger} )^{\otimes 2} ] \bigr)
=
\frac{1}{d_{A_2} d_{A_1}^4}
\biggl[
d_{A_2} I_{A_1 R}^{\otimes 2}
+
\mathbb{F}_{A_1 A_1'} \otimes \mathbb{F}_{RR'}
+
\mathbb{L}^X_{A_1 A_1'} \otimes \mathbb{L}^Z_{R R'}\\
-
 \mathbb{L}^X_{A_1 A_1'} \otimes (I_{RR'}+\mathbb{F}_{RR'} )
-
d_{A_2} I_{A_1}^{\otimes 2} \otimes \mathbb{L}^X_{RR'}
-
\mathbb{F}_{A_1 A_1'} \otimes \mathbb{L}^Z_{RR'}\\
+
\sum_{i,j=1}^{d_{A_1}} \sum_{\alpha, \beta=1}^{d_{A_1}}
i_{\alpha}i_{\beta}j_{\alpha}j_{\beta}
\bigl(
d_{A_2} \ketbra{\alpha \beta}{\alpha \beta}_{A_1A_1'} \otimes \ketbra{ij}{ji}_{RR'}
+ \ketbra{\alpha \beta}{\beta \alpha }_{A_1A_1'} \otimes \ketbra{ij}{ij}_{RR'}
\bigr)
\biggr],
\end{multline}
where $\mathbb{L}^Z = \sum_i \ketbra{ii}{ii}$, $\mathbb{L}^X$ is the Hadamard conjugation of $\mathbb{L}^Z$, $\{ \ket{\alpha} \}$ is the Pauli-$X$ basis, and $i_{\alpha} = \braket{i}{\alpha}$.
Substituting this into Eq.~\eqref{Eq:12345}, and using $\mc{H}_R \cong \mc{H}_{A_1}$ and relations $\tr \mathbb{F}_{A_1A_1'}=\tr \mathbb{L}_{A_1A_1'}^W \mathbb{F}_{A_1A_1'} = d_{A_1}$ for $W=X,Z$, we obtain
\begin{equation}
\mathbb{E}_{U_A \sim {\sf D}_X \otimes {\sf D}_Z} [\tr \bigl( \mc{T}_{A \rightarrow A_1} (U_A  \rho_{AR}  U_A^{\dagger} ) \bigr)^2]
=
\frac{1}{d_{A_1}} + \frac{1}{d_{A_2}}- \frac{1}{d_{A}}.
\end{equation}
Substituting this into Eq.~\eqref{Eq:last}, we conclude the proof. 

>From Theorem~\ref{Thm:decs}, it is straightforward to calculate the decoupling rate with a Haar random unitary by taking $\epsilon \rightarrow 0$.
$\hfill \blacksquare$

\section{Applications in detail} \label{App:decApp}

In this section, we provide more detailed discussions about applications of decoupling theorems given in Section~\ref{MainResult}.

\subsection{Information theoretic protocols} \label{Appp:ITP}
Here, we consider the coherent state merging protocol in a one-shot setting~\cite{ADHW2009, B2009, DH2011}. Out of all information theoretic protocols this one is of particular interest, since many others can be derived from it as special cases. The task here is to transfer some part $A$ of a tripartite quantum state $\Psi^{ABR}$ from one party to an other which might already possess a part $B$ of the state, while leaving the reference $R$ unchanged. At the same time both parties try to generate as much shared entanglement as possible. In a more precise way the protocol is defined as follows. 
\begin{definition}
Consider a pure tripartite quantum state $\Psi^{ABR}$ where $A=A_1A_2$ and $B=B_1B_2$. Now, coherent state merging with error $\epsilon>0$ is a LOCC process $\mathcal M^{AB \rightarrow A_1B_1\tilde B BR}$, with
\begin{equation}
|\!| \mathcal M^{AB \rightarrow A_1B_1\tilde B BR}(\Psi^{ABR}) - \Phi^{A_1B_1} \otimes \Psi^{\tilde BBR} |\!|_1 \leq \epsilon, 
\end{equation}
where $\Phi^{A_1B_1}$ is a maximally entangled state and $\Psi^{\tilde BBR} = {\rm id}^{A\rightarrow \tilde B} \Psi^{ABR}$. The entanglement gain is now given as $e_\epsilon^{(1)} = \log |A_1|$ and the communication cost by $q_\epsilon^{(1)} = \log |A_2|$. 
\end{definition}

In the general one-shot setting, the optimal rates for coherent state merging have been shown \cite{DH2011} to be 
\begin{align}
e_\epsilon^{(1)} &\geq \frac{1}{2}[H_{\rm min}^{\delta}(A|R)_\Psi +H^\delta_0(A)_\Psi]
+ \log(\delta' ) \\
q_\epsilon^{(1)} &\leq \frac{1}{2}[-H_{\rm min}^{\delta}(A|R)_\Psi +H^\delta_0(A)_\Psi]
- \log(\delta')
\end{align}
for $0 < \epsilon\leq 1$ and $\delta >0$ such that $\epsilon = 2\sqrt{5\delta'} + 2\sqrt{\delta}$, $\delta' = \delta + \sqrt{4\sqrt{\delta}-4\delta}$ and $H_0(A)_\rho = \log{\tr{\Pi_{\rho^A}}}$,where $\Pi_{\rho^A}$ denotes the projector onto the support of $\rho^A$. Here, we use the term \emph{optimal} as done in Ref.~\cite{DH2011}, where it was shown that these rates indeed converge to the optimal ones in the limit of many identical copies of the state $\Psi$. Furthermore the same work also gave a converse bound in the one-shot case. In Ref.~\cite{DH2011} the achievability of these rates was shown using a decoupling theorem with a Haar random unitary, implying that there always exists at least one that can be used for sufficiently precise decoding. 
It is natural to ask whether similar rates can be achieved using simpler unitaries.
This was, for example, investigated in Ref.~\cite{HM14} where the optimal rates for coherent state merging were calculated using a decoupling theorem based on an approximate unitary $2$-design. 

Our result offers the opportunity to implement the encoding for coherent state merging using a unitary in $\mc{D}[\ell]$, with achievable rates depending on the number of repetitions. To state these rates we note that we use the following corollary, obtained by rewriting Eq.~\eqref{Eq.epsdecDU}, when the CPTP map is given by the partial trace over a subsystem of $A$.

\begin{corollary}[Smooth decoupling with $\mathbf{D[\ell]}$ for the partial trace] \label{Thm:decPT2}
Let a system $A$ be composed of two subsystems $A_1$ and $A_2$, of which Hilbert spaces are $\mc{H}_{A_1}$ and $\mc{H}_{A_2}$ with dimension $d_{A_1}$ and $d_{A_2}$, respectively. 
While 
\begin{equation}
\log{d_{A_1}}\leq \frac{1}{2}( H^{\epsilon}_{\rm min}(A|R)_{\rho} +\log{d_A})+\log \epsilon + \log(1+8 d_A^{2-\ell}) , \label{Eq:23m'o2}
\end{equation} 
the following holds
\begin{equation}
\mathbb{E}_{U_A \sim {\sf D}[\ell]}|\!| \tr_{A_2} (U_A \rho_{AR} U_A^{\dagger}) - \tau_{A_1} \otimes  \rho_R |\!|_1 \leq 9\epsilon
\end{equation}
to the leading order.
\end{corollary}

Using this, we can express the rates for coherent state merging using random diagonal unitaries.
\begin{theorem}[Achievable rates for coherent state merging with $\mathbf{D[\ell]}$]
Fix $0 < \epsilon\leq 1$. Given a tripartite pure state $\Psi^{ABR}$, the quantum communication cost $q_\epsilon^{(1)}$ and the entanglement gain $e_\epsilon^{(1)}$ of any $\epsilon$-error one-shot coherent state merging protocol with $D[\ell]$ satisfies the following bounds:
\begin{align}
e_\epsilon^{(1)} &\geq \frac{1}{2}[H_{\rm min}^{\delta}(A|R)_\Psi +H^\delta_0(A)_\Psi]
+ \log(\delta') + \log(1+8 d_A^{2-\ell}) \\
q_\epsilon^{(1)} &\leq \frac{1}{2}[-H_{\rm min}^{\delta}(A|R)_\Psi +H^\delta_0(A)_\Psi]
- \log(\delta') - \log(1+8 d_A^{2-\ell})
\end{align}
for $\delta >0$ such that $\epsilon = 2\sqrt{9\delta'} + 2\sqrt{\delta}$, and $\delta' = \delta + \sqrt{4\sqrt{\delta}-4\delta}$.
\end{theorem}

We can conclude that coherent state merging can be achieved with essentially optimal rates for $\ell\geq 2$ when using our decoupling technique and therefore with an efficient implementation. In the line of Ref.~\cite{HM14} this results further improves the decoding complexity by showing that even simpler implementations then those using unitary $2$-designs can achieve the given rates. 

\subsection{Relative thermalisation} \label{Appp:RThm}
We here discuss implications of our results on relative thermalisation, which is one of the fundamental questions in thermodynamics.

Formally, relative thermalisation is defined as follows.
\begin{definition}[Ref.~\cite{dRHRW2014}]
Let $S,E,R$ be quantum systems and $\Xi \subseteq S\otimes E$ be a subspace with a physical constraint such as total energy. The global system is in a state $\rho_{\Xi R} \in \mc{S}(\Xi \otimes R)$. $S$ is called \emph{$\delta$-thermalised relative to $R$} if
\begin{equation}
|\!| \rho_{SR} - \pi_S\otimes\rho_R |\!|_1 \leq\delta ,
\end{equation}
where $\rho_R$ is arbitrary and $\pi_S$ is a local microcanonical state, defined as $\pi_S := \tr_E \pi_\Xi$, where $\pi_\Xi := \frac{I_\Xi}{|\Xi |}$ is the fully mixed state on $\Xi$.
\end{definition}
Here, the $S$, $E$, and $R$ should be identified with the system, an environment, and a reference, respectively. With this identification, relative thermalisation requires that, after the environment is traced out, the system should be close to the state $\pi_S$ and should not have strong correlations with the reference.
When the physical constraint defining the subspace $\Xi$ is given by total energy, $\pi_S$ is simply a thermal Gibbs state on $S$. 
%Hence, the relative thermalisation means that $S$ is not only thermalised but also has no correlation with $R$. 
Thus, relative thermalisation is considered to be a mathematical definition of one of the assumptions in standard thermodynamics that \emph{the system should lose all correlations with other systems once it becomes a thermal Gibbs state}.
% In the special case of $d_R =1$, this reduces to the well-known canonical typicality \cite{GLTZ2006,  PSW2006, R2008}. 

It is important to clarify when and how relative thermalisation occurs in physically feasible situations. One idea is, as proposed in Ref.~\cite{dRHRW2014}, to apply decoupling theorem
by assuming that Haar random unitaries are acting on $\Xi$. It is however not clear whether Haar distributed unitaries have any physical interpretation.  
Since the random unitary $D[\ell]$ may typically appear as a natural Hamiltonian dynamics in certain types of quantum chaotic systems (see Ref.~\cite{NHMW2015-1, NHKW2016}), our decoupling result provides a relative thermalisation in physically feasible systems. 

We first rewrite Theorem~\ref{Thm:CoMtorus} in the following form. 

\begin{corollary}[Probabilistic decoupling theorem with $\mathbf{D[\ell]}$] \label{Thm:CoM}
Let $\epsilon, \Delta, \eta >0$. If the entropic relation 
\begin{equation}
H^{\epsilon}_{\rm min}(A|E)_{\rho}+H^{\epsilon}_{\rm min}(A|B)_{\tau} \geq 2 \log{\frac{2 \sqrt{1+8 d_A^{2-\ell}}}{\Delta - 12\epsilon}}
\end{equation}
holds, then the fraction of unitaries $U_A \in \mc{D}[\ell]$ such that 
\begin{equation}
|\!| \mc{T}_{A \rightarrow B} (U_A \rho_{AE} U_A^{\dagger}) - \tau_B \otimes  \rho_E  |\!|_1 
\geq \Delta + \eta
\end{equation}
is at most $2 e^{-\chi_{\ell} d_A \eta^4}$ with respect to the measure ${\sf D}[\ell]$, where $\chi_{\ell}^{-1} = 2^{11} (2 \ell +1)^3  K^4 \pi^2$, and $K=d_A |\!| \rho_A |\!|_{\infty}$.
\end{corollary}

Using Corollary~\ref{Thm:CoM} and setting $A= \Xi$, $B=S$, and $\mc{T}_{A \rightarrow B} = \tr_E$, we can calculate the fraction of relative thermalisation with random diagonal unitaries.
Following the approach in Ref.~\cite{dRHRW2014}, we obtain the following theorem.

\begin{theorem}[Relative thermalisation using $\mathbf{D[\ell]}$] \label{Thm:RelThe}
Let $\rho_{\Xi R}\in \mathcal S(\Xi \otimes R)$, with $\Xi\subseteq S\otimes E$. Let $\epsilon_2, \epsilon_3, \Delta >0$ and $\epsilon_1 > \epsilon_2 + \epsilon_3$. If the entropic relation
\begin{equation}
H^{\epsilon_1}_{\rm min}(SE|R)_{\rho}+H^{\epsilon_2}_{\rm min}(E)_{\pi} - H^{\epsilon_3}_{\rm max}(S)_{\pi}\geq 2 \log{\frac{2 \sqrt{1+8 d_S^{2-\ell}}}{(1-\sqrt{1-(\epsilon_1 - \epsilon_2 - \epsilon_3)^2})(\Delta - 24\epsilon_1)}}
\end{equation}
holds, then, after a unitary evolution $U \in \mc{D}[\ell]$ acting on $\Xi$, $S$ will be $\delta$-thermalised relative to $R$, except for a fraction of at most $2 e^{-\chi_{\ell} d_S \delta^4}$ of the unitaries with respect to ${\sf D[\ell]}$, where $\chi_{\ell}^{-1} = 2^{11} (2 \ell +1)^3  K^4 \pi^2$, and $K=d_S |\!| \rho_S |\!|_{\infty}$.
\end{theorem}

Theorem~\ref{Thm:RelThe} provides the conditions when relative thermalisation occurs, which is similar to Ref.~\cite{dRHRW2014}. 
The fraction is however dependent on the constant $K=d_S |\!| \rho_S |\!|_{\infty}$, which satisfies $1 \leq K \leq d_S$. When $K = \Omega(d_S^{1/4})$, the fraction bound in Theorem~\ref{Thm:RelThe} becomes rather trivial.
However, we are especially interested in the situation where initial state $\rho_{\Xi R}$ is strongly correlated between $S$ and $R$ because the main goal is to understand whether the correlations are destroyed during the thermalisation process.
In this case, $K$ is expected to be relatively small and hence, Theorem~\ref{Thm:RelThe} shall imply that almost any dynamics in sufficiently chaotic many-body systems, modelled by $D[\ell]$, makes the system $S$ thermalised relative to $R$.

\bibliographystyle{unsrt}

\bibliography{Bib}

\end{document}